\documentclass{ecai} 

\usepackage{latexsym}
\usepackage{amssymb}
\usepackage{amsmath}
\usepackage{amsthm}
\usepackage{booktabs}
\usepackage{enumitem}
\usepackage{graphicx}
\usepackage{color}
\usepackage{subfigure}
\usepackage{hyperref}
\usepackage[ruled,vlined]{algorithm2e} 

\usepackage{amsmath}
\newtheorem{lemma}{Lemma}

\newtheorem{theorem}{Theorem}

\newtheorem{definition}{Definition}

\newcommand{\BibTeX}{B\kern-.05em{\sc i\kern-.025em b}\kern-.08em\TeX}

\begin{document}


\begin{frontmatter}

\paperid{123} 


\title{A Dynamic Service Offloading Algorithm Based on Lyapunov Optimization in Edge Computing}

\author[A]{\fnms{Peiyan}~\snm{Yuan}}
\author[B]{\fnms{Ming}~\snm{Li}}
\author[C,D]{\fnms{Chenyang}~\snm{Wang}\thanks{Corresponding Author. Email: chenyangwang@ieee.org}} 
\author[D]{\fnms{Ledong}~\snm{An}} 
\author[A]{\fnms{Xiaoyan}~\snm{Zhao}} 
\author[A]{\fnms{Junna}~\snm{Zhang}} 
\author[E]{\fnms{Xiangyang}~\snm{Li}} 
\author[F]{\fnms{Huadong}~\snm{Ma}} 

\address[A]{Henan Normal University}
\address[B]{Xidian University}
\address[C]{Shenzhen University}
\address[D]{Guangdong Laboratory of Artificial Intelligence and Digital Economy (SZ)}
\address[E]{University of Science and Technology of China}
\address[F]{Beijing University of Posts and Telecommunications}




\begin{abstract}
This study investigates the trade-off between system stability and offloading cost in collaborative edge computing. While collaborative offloading among multiple edge servers enhances resource utilization, existing methods often overlook the role of queue stability in overall system performance. To address this, a multi-hop data transmission model is developed, along with a cost model that captures both energy consumption and delay. A time-varying queue model is then introduced to maintain system stability. Based on Lyapunov optimization, a dynamic offloading algorithm (LDSO) is proposed to minimize offloading cost while ensuring long-term stability. Theoretical analysis and experimental results verify that the proposed LDSO achieves significant improvements in both cost efficiency and system stability compared to the state-of-the-art.
\end{abstract}

\end{frontmatter}


\section{Introduction}

The rapid proliferation of fifth-generation (5G) mobile communication technology has significantly increased the number of network-connected devices and enhanced society's perceptive capabilities~\citep{duan2018treasure}.
However, conventional cloud computing architectures struggle to meet the stringent requirements for large-scale, real-time data processing and energy efficiency. 
Meanwhile, backbone networks are experiencing unprecedented resource constraints~\citep{chen2024joint}. 
To address these challenges, edge computing has emerged as a promising paradigm, deploying edge servers in close proximity to data sources to provide services such as content caching, computation, and distribution. This paradigm effectively alleviates core network congestion and reduces user-perceived latency~\citep{wang2024dependency,wang2023heterogeneous,wang2023mission}.

Service offloading in edge computing generally falls into two categories: horizontal and vertical~\citep{r11}. 
Horizontal offloading refers to collaborative processing among multiple edge servers, commonly termed multi-edge server cooperative offloading. Vertical offloading, on the other hand, involves coordinated decision-making between edge servers and cloud servers, often via negotiation mechanisms~\citep{r12,r14,r15,yu2022deterministic}. 
The vertical approach is suitable for computation-intensive tasks that benefit from the cloud’s vast processing capacity, whereas the horizontal approach is more effective for distributed processing of large-scale data with stringent latency requirements~\citep{sun2024hierarchical}. 
In recent years, the horizontal paradigm has attracted increasing research attention due to its potential for low-latency, high-throughput service delivery~\citep{r18,chen2022dynamic,chen2021multitask}.
	

Despite the advantages of horizontal cooperation, dynamic and heterogeneous network environments introduce significant challenges. Fluctuations in service requests and task arrival rates hinder the timely adaptation of traffic distribution among edge nodes~\citep{r23,r24,r25}. 
Moreover, processing resource-intensive tasks may cause buffer accumulation and link congestion, resulting in system instability, increased response times, and higher energy consumption~\citep{r1234,r30,r31}. 
These issues highlight the necessity of designing efficient service offloading schemes that jointly consider resource availability, workload variations, and system stability. 
Striking an optimal balance between dynamic traffic control and overall system cost in multi-edge cooperative scenarios remains a challenging and unresolved problem.
	
	
This study investigates the minimum offloading cost problem of edge server collaborative offloading and further transforms it into a stochastic optimization problem related to delay and energy consumption. 
In order to achieve the long-term optimal solution between computing and communication costs in the edge offloading system, the Lyapunov optimization framework is employed to limit the queue length of edge servers through real-time offloading decisions. 
The main contributions of this study are summarized as follows.
\vspace{-0.5cm}
\begin{itemize}	
    \item A time-varying queue model is proposed to reflect the dynamic state of the buffers. In addition, a system cost model that encompasses factors of energy consumption and delay is formulated, associated with data transmission, receiving, and processing costs.
    \item The Lyapunov function cluster (\textbf{\emph{Quadratic, Drift and Drift-Plus-Penalty}}) is used to balance the buffered data volume and system cost. An upper bound on the system cost and the stability of queues are rigorously analyzed in theory, respectively.
    \item A greedy matching algorithm is proposed to select the qualified edge servers, and the convergence and stability of the algorithm are proven. In addition, a large number of experimental results verify the asymptotic optimally of the proposed algorithm.
\end{itemize}	

\section{Related Work}

This section introduces related work from two aspects of vertical collaboration and horizontal collaboration.
Vertical collaborative offloading can make full use of the powerful computing resources of the cloud center. 
It is more suitable for applications with strict requirements on computing resources rather than those sensitive to delay, such as training offline machine learning models.

The authors of~\citep{r32} used Little's law and the queue length to evaluate delay, but the resource availability of edge servers is ignored. Literatures of~\citep{r35,r36} took energy consumption, delay and cost as the performance index to offloading services respectively, which are not suitable for some complex and real-time offloading scenarios. References~\citep{r39,r44,r45} considered the system delay and cost,but ignored the energy consumption per bit of data in cloud-edge systems with dynamic workloads. 
All the above studies focus on improving the performance of offloading schemes/algorithms, failing to consider the limited service storage and data processing capabilities of edge servers under long-term constraints. Meanwhile, most of these works do not consider the collaboration among multiple edge servers.
	
Horizontal collaborative offloading can take advantage of the resources of multiple edge servers, and process user or terminal requests with lower latency and higher efficiency. 
In horizontal cooperative offloading, it is a challenging problem to make a balance between the stability of the edge offloading system and the offloading performance. 
For example, the authors of~\citep{r43,r47,r49} focused on terminal queues, which include vehicles, wearable devices and smartphones, etc. Works in~\citep{r26,r32,r36} evaluated the stability of single edge server, but they do not consider the dynamic data flow changes and long-term constraints between multiple edge servers. Literatures~\citep{r23,r37,r38,r40,r41,r43,r44} created queues for multiple edge servers, and considered the trade-off between power consumption and execution delay in multi-user/device MEC systems. However, the long-term system stability is ignored.

In general, the works~\citep{r44,r43,r41} all considered the cost, but some works considered a single cost optimization objective~\citep{r37,r38,r40,r45}, while others considered both cost and other performance indicators.  
For example, from an economic point of view, \citep{r37} considered the cost of renting data resources of the edge server for data caching and the migration cost of data transmission to the edge server, and obtained the upper limitation of the cost through Lyapunov drift optimization. 
The authors of~\citep{r40} considered the cost of data uploading and processing of multiple time slices from an economic point of view, and solved the cost through Lyapunov optimization and distributed Markov approximation methods. 
From the perspective of queues, \citep{r38} constructed and maintained a virtual queue of service migration cost for each base station, where the cost includes two parts: dynamic deployment of services and additional migration. The trade-off between migration cost and system stability was obtained through Lyapunov optimization and Chernoff boundary theory.  
The cost composition of included three parts: service replacement, storage, and maintenance in~\citep{r45}. 
Through Lyapunov drift and regularization, the non-convex optimization problem was transformed into convex optimization, and then the cost was solved.
However, the above studies did not consider the energy consumption caused by each bit of data. 
	
\section{Architecture of Edge Offloading System}
The system model is shown in Figure~\ref{fig1}. The edge service offloading system architecture includes one macro base station (BS), $M$ small base stations (SBS), and $n$ users. $M$ SBSs are covered by the BS. Unlike most existing single-hop service offloading, tasks can be finished by the cooperative edge server, and multiple SBSs jointly provide services to users through horizontal cooperation. 
When users send a service request to a nearby SBS, if the SBS deploys the user's requested service, it processes the service locally. Otherwise, the edge server sends a request to other edge servers for help.
	
Users and SBSs follow the PPP distribution with intensity $\lambda_{ u}$ and $\lambda_{e}$, respectively. Assuming that the BS can access a database $F=\left \{ F_{1}, F_{2},..., F_{K } \right \}$ with $K$ independent services, since different services require different storage space, the size of the $k$-th type of service is $b_{k}$ bits and the caching capacity of each SBS is $B_{i}$ bits. Assume that there are $a_{k}$ total of SBSs that provide the $k$-th type of service, and $f_{i}^{k}=\left \{ f_{1}^{k},f_{2}^{k},...,f_{M}^{k} \right \}$ is the service placement vector, indicating whether the $k$-th service data is cached in the $i$-th SBS. Therefore,  ${\textstyle \sum_{i=1}^{M}}f_{i}^{k}(t)=a_{k}$ is the number of providers of the $k$-th type of service in the system.
\begin{figure}[t]
    \includegraphics[width=3in,height=2in]{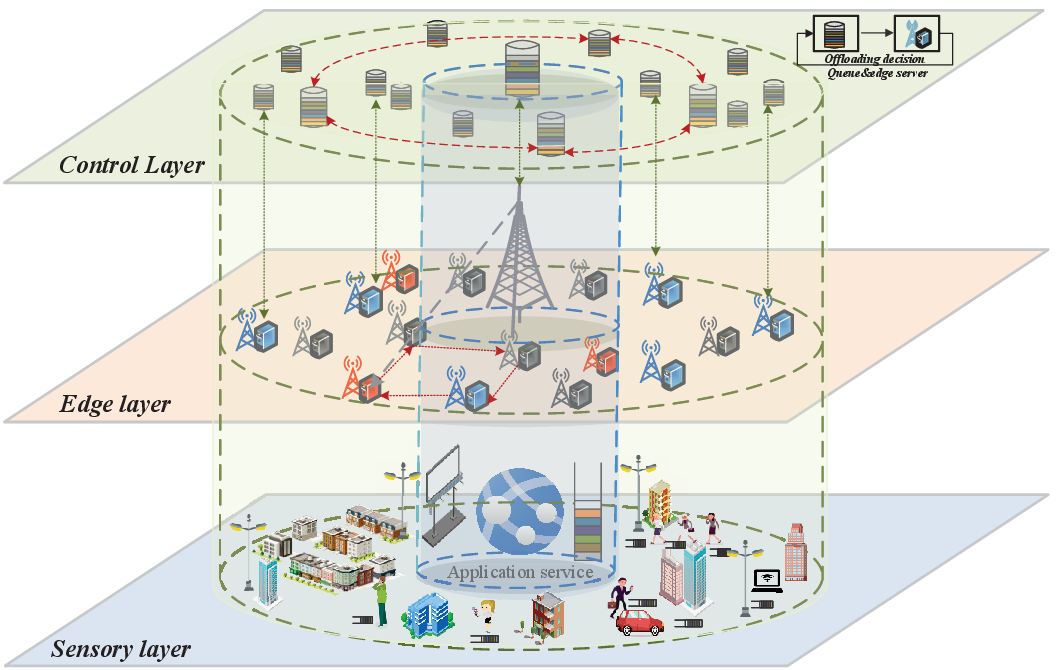}
\caption{Edge service offloading architecture.}
\label{fig1}	
\vspace{0.39cm}	
\end{figure}
When users in the system send requests for the $k$-th service to nearby edge servers, $x_{i}^{k}(t)=1$ represents the processing of the $k$-th type of service in the $i$-th edge server. Conversely, 0 represents that the $i$-th edge server cannot process it, and the service request needs to be collaboratively processed by multiple edge servers and the cloud center.

\section{Problem Modeling and Analysis}
\subsection{Problem statement}
To reduce the long-term edge offloading system's cost $Cost(t)$, the problem can be formulated as follows.
\begin{align} 
    {\min} &\  \lim _{T \rightarrow\infty} \frac{1}{T} \sum_{t=0}^{T-1} E\{\operatorname{Cost}(t)\} \label{eq1}\\ 
    s.t.~ 
    &  \bar{Q}=\lim _{T \rightarrow \infty} \frac{1}{T} \sum_{t=0}^{T-1} \sum_{i=1}^{M} E\left\{Q_{i}(t)\right\}<+\infty
    \label{eq2} \tag{1a} \\ 
    &\forall i \in[1, M], Q_{i}(t) \leq Q_{i}^{\max}
    \label{eq3}  \tag{1b}\\ 
    & \sum_{k=1}^{K} n_{i}^{k}(t) b_{k} \leq\left(Q_{i}^{\max }-Q_{i}(t)\right)\le \mu_{i}(t), \forall i \in[1, M]
    \label{eq4} \tag{1c}\\ 
    & \forall i \in[1, M], \mu_{i}(t) \leq \mu_{i}^{\max}
    \label{eq5} \tag{1d} \\
    & \forall t \in[0, T-1],E^c(t)+E^p(t) \leq E^{\max} \tag{1e} \label{eq6-1}\\
    & T^c(t)+T^p(t) \leq T^{\max} \tag{1f} \label{eq6-2}\\
    & \forall k\in\left [1, K \right], \sum_{i=1}^{M} x_{i}^{k}(t)=1
\label{eq7} \tag{1g} 
\end{align}	
Constraint (\ref{eq2}) represents the stability of the buffer queue, (\ref{eq3}) and (\ref{eq4}) indicate the constraints of the buffered data and the processing ability of the edge server, respectively. 
When the pending service data reaches the maximum, the delay increases, and a large amount of data may be lost. 
Constraint (\ref{eq5}) represents that the processing capacity of each edge server must not exceed its maximum value. 
Constraints (\ref{eq6-1}) and (\ref{eq6-2}) ensure that the energy consumption and delay of the system do not exceed their maximum value in each time slot. (\ref{eq7}) represents the 0-1 offloading constraint of the service, where each type of service can only be offloaded to one edge server. In the following section, the queue model will be analyzed in detail.
	
This study investigates energy consumption in data transmission and processing, which inevitably involves computation and communication delays. Therefore, a system cost model including energy consumption and delay is constructed. At the same time, it is reasonable to assume that IoT terminals have little computing capability, so terminal computing is not considered, \emph{i.e.,}
\begin{equation}
    \operatorname{Cost}(t)=\theta \cdot E_{total}(t)+(1-\theta )\cdot T_{total}(t)
\label{eq8}
\end{equation}
where $\theta$ is a weight factor that adjusts the ratio of energy consumption to delay in the cost and $t\in \left \{0,1,2,..., T-1 \right \}$.  $E_{total} (t)$ and $T_{total} (t)$ are the total energy consumption and delay, respectively.
	
When users send service requests to an edge server, the offloading decision $x_{i}^{k}(t) =  \{ 0, 1\}$, where 1 represents that the $k$-th type of service is offloaded to the adjacent edge server $i$ at time slot $t$, where $i \in \left \{1,2,..., M \right \} $, and conversely, 0 represents that the edge server cannot process the request, and the service offloading is performed through multi-hop collaboration.
	
This study uses a time-varying service request probability model~\cite{r50} to calculate the probability of the $k$-th service type, rather than the traditional Zipf distribution. The number of users requesting the $k$-th type of service at time $t$ is represented as $n_{i}^{k}(t)=n_{i}\times P_{k}(t)$,
where $P_{k} (t)$ represents the probability of users requesting the $k$-th service at time $t$, $n_{i}$ is the number of users covered by the $i$-th edge server. The amount of service data received by the $i$-th edge server in the time slot is thus represented as follows.
\begin{equation}
    A_{i}(t)=\sum_{k=1}^{K} n_{i}^{k}(t) b_{k}
\label{eq10}
\end{equation}
Furthermore, we define $A_{i}^{max}$ as the maximum value received by the $i$-th edge server, and we have $A_{i}(t)\le A_{i}^{max}$.

\subsection{Problem Analysis}
\subsubsection{Buffer queue model of edge servers}

Note that each edge server only deploys part of the services, and it is impossible to finish all requests within one time slot. It is necessary to consider the buffer queue model to ensure the stability of the edge offloading system. The buffer queue model is shown in Figure 2, which is maintained for each edge server to store the unfinished service data in the last time slot.
\begin{figure}[t]
\includegraphics[width=3.5in,height=2in]{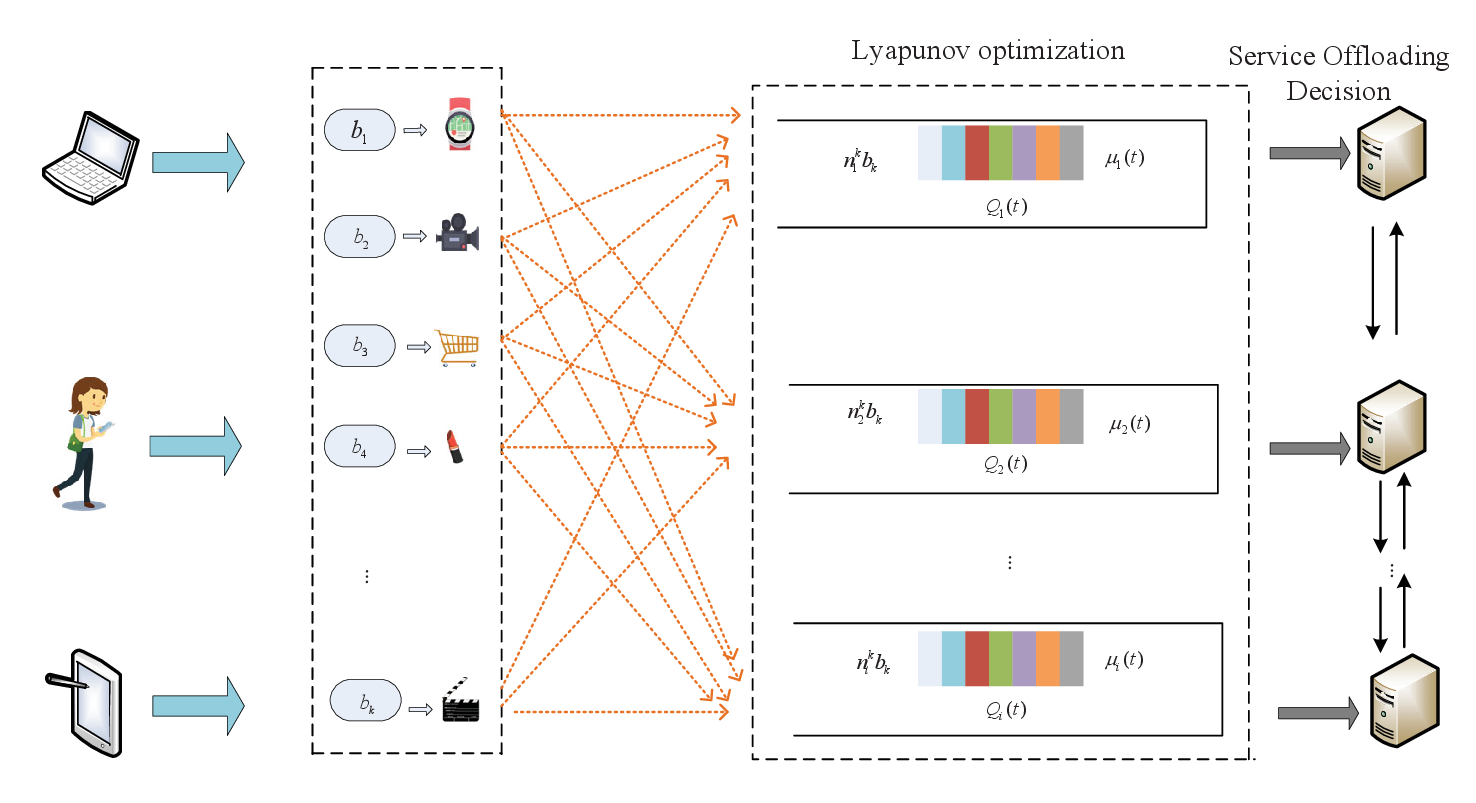}
\caption{Edge server buffer queue model.}	
\vspace{0.4cm}	
\end{figure}
	
Let $Q_{i}^k(t)$ denote the amount of data from the $k$-th service in the buffer of the $i$-th edge server. The buffer occupation from all kinds of services can be presented as follows.
\begin{equation}
    Q_{i}(t)=\sum_{k=1}^{K} Q_{i}^k(t)
\label{eq11}
\end{equation}
	
The amount of buffered data in two adjacent time slots is 
\begin{equation}
    Q_{i}(t+1)=\max \left\{Q_{i}(t)-\mu_{i}(t), 0\right\}+A_{i}(t)
\label{eq12}
\end{equation}
where $\mu_{i}(t)$ represents the processing rate of the $i$-th edge server, and Equation~(\ref{eq12}) states that the service data volume offloaded in time slot $t$ will be stored in the queue and processed in the next time slot. If the amount of data processing in a time slot is less than the amount of data received.
Let $\mu_{i}^{max}$ represent the maximum workload of the edge server. 
At each time slot, the amount of processed and completed data is less than the maximum workload, \emph{i.e.,} $\mu_{i}(t)\le \mu_{i}^{max}$.
	

Specifically, the buffer queue length must be carefully managed to avoid unbounded growth, which would otherwise lead to excessive delays in data processing. 
To address this, the average buffer queue length should be constrained within acceptable limits. 
This implies that the queue length at time $t$ must remain below a predefined upper threshold, \emph{i.e.,} $Q_{i}(t)\le Q_{i}^{max}$.

\subsubsection{Cost Model of Edge Offloading System}
\paragraph{Energy consumption model}
System energy consumption is generated by data transmission in communication and data processing in edge server. In this study, the energy consumption in data packet forwarding and receiving is considered as part of the transmission energy consumption. Define $E^{c}(t)$ as energy consumption of communication and $E^{p}(t)$ as energy consumption of data processing, the total energy consumption $E_{total}(t)$ can be represented as follows.
\begin{equation}
    E_{total}(t)=E^{c}(t)+E^{p}(t)
\label{eq13}
\end{equation}
	
The energy consumption generated within each phase can be named as device-server and server-server energy consumption, respectively. Define the device-server energy consumption as $E^{c1}(t)$ and the server-server energy consumption as $E^{c2}(t)$, the total energy consumption of data transmission is
\begin{equation}
    E^{c}(t)=E^{c1}(t)+E^{c2}(t).
\label{eq14}
\end{equation}
	
The energy consumption of data transmission depends on the transmission power and the amount of transmitted data.
Let $e_{u}$ represent the unit energy consumption generated by users, then the device-server energy consumption $E^{c1}(t)$ is
\begin{equation}
    E^{c1}(t)=\sum_{i=1}^{M} \sum_{k=1}^{K}n_{i}^{k}(t)b_{k}e_{u}
\label{eq15}
\end{equation}
Equation~(\ref{eq15}) shows that the energy consumption has a linear relationship with the amount of data received, that is, the energy consumption will grow accordingly as the workload increases.
	
On the other hand, assume that the number of edge servers in the collaboration area is represented as $O_{k} (t)$. When the service data is transmitted through multiple hops, the maximum hops can be written as $H_{k} (t)=O_{k} (t)-1$. Thus, the energy consumption for receiving and forwarding data through collaboration among servers is
\begin{equation}
    E^{c2}(t) = \sum\limits_{i = 1}^M {\sum\limits_{k = 1}^K {{n_{i}^{k}}}} (t){b_k}e_s H_{k} (t)\left({1 - x_i^k(t)} \right).
\label{eq16}
\end{equation}
	
Assume that there are K types of services, and each SBS can only cache part of the types of services. When an SBS receives a request, the probability for the SBS to process is $P_{c} ={\bar A_{i}}/K$, where ${\bar A_{i}}$ denotes the number of service types cached in the $i$-th SBS. When the service request is forwarded to other cooperating SBSs, the probability that the service is responded to by other SBSs is $P_{nc} =\left (K-{\bar A_{i}}\right )/K$.

\begin{figure*}[th]
\centering{
    \includegraphics[width=0.7\textwidth]{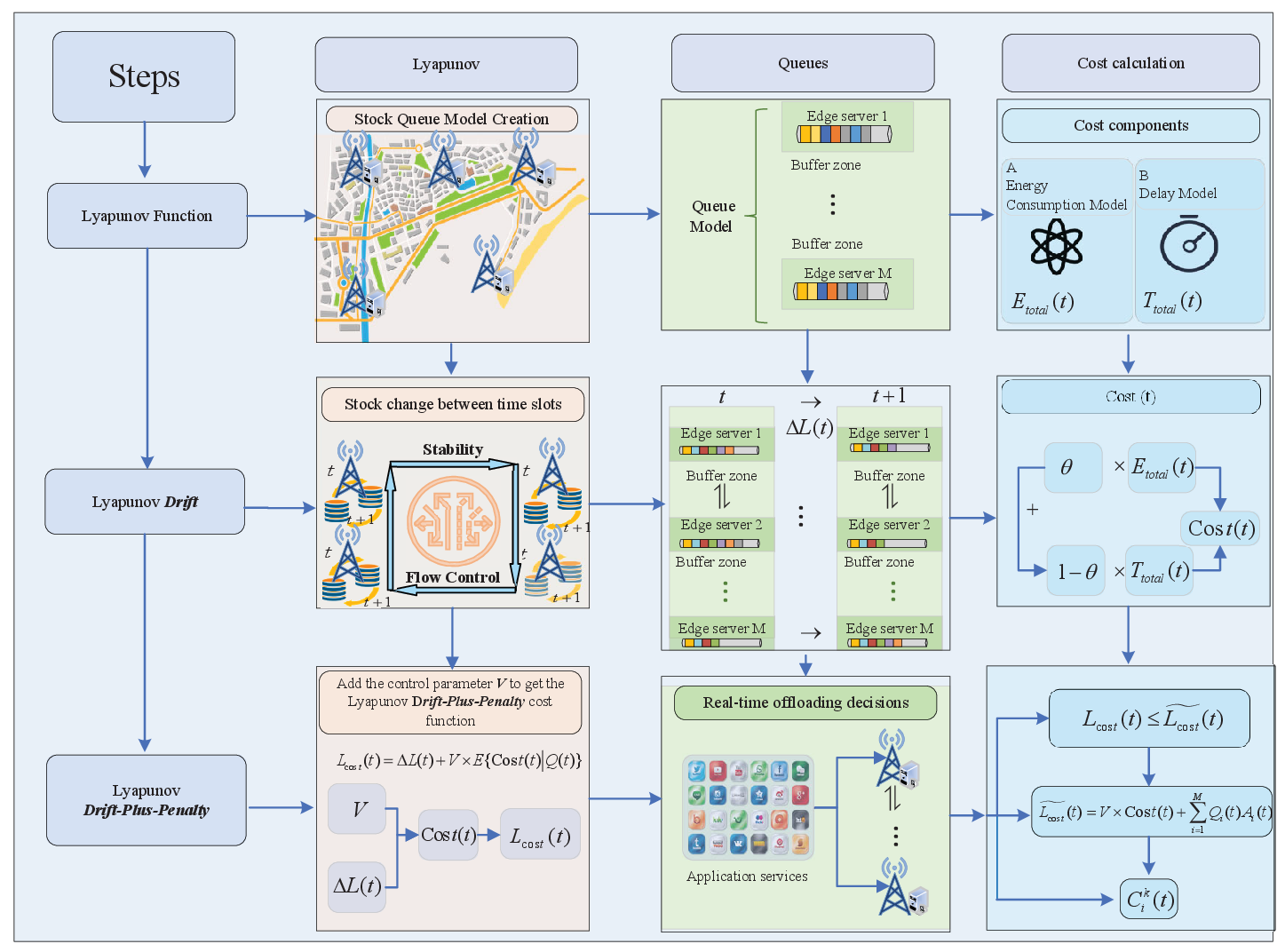}}
\caption{LDSO algorithm framework}	
\end{figure*}	
	
Let $e_{p}$ denote the processing power per bit of data of the edge server. The energy consumption of data process $E^{p}(t)$ is
\begin{equation}
    {E^p}(t) = \sum\limits_{i = 1}^M {\sum\limits_{k = 1}^K {{e_p}{A_i}(t)\left( {{P_c}x_i^k(t) + {P_{nc}}\left( {1 - x_i^k(t)} \right){H_k}(t)} \right)}}.
\label{eq17}
\end{equation}

\paragraph{Delay model}
Delay in the service offloading process includes communication delay $T^{c}(t)$ and computation delay $T^{P}(t)$. Therefore, the total delay is shown as
\begin{equation}
    T_{total}(t)=T^{c}(t)+T^{p}(t).
\label{eq18}
\end{equation}

The communication delay from the user device to the $i$-th edge server is related to the amount of offloaded data $n_{i}^{k}(t)b_{k}$, that is, the product of the number of requesting users for the $k$-th type service and the size of the service. Assume that the communication bandwidth is $B_{w}$, the device transmit power is $P_{u}$ and the channel gain is expressed as $h(t)$. Let $r(t)$ be the channel transmission rate, the communication delay can be expressed as:
\begin{equation}
    T^c(t)=\sum_{i=1}^M \sum_{k=1}^K \frac{n_{i}^{k}(t) b_k}{r(t)}
\label{eq19}
\end{equation}
where $r(t)=B_w \cdot \log _2\left(1+\frac{P_{u}h(t)}{\varphi^2}\right)$.
	
To describe the computation delay of the service, each SBS that contributes to the service data offloading of the $k$-th type service is regarded as a M/M/1 queue. According to queuing theory, the arrival flow can be modeled as Poisson distribution. Let $\lambda_{u}^{k}$ denote the arrival intensity of the $k$-th service from the $u$-th user. The average wait time can be represented as
\begin{equation}
    t_k=\frac{1}{\mu_{i}(t)-\sum_{u=1}^{n_{i}^{k}(t)} \lambda_{u}^{k} b_k}.
\label{eq20}
\end{equation}
	
As described in the service offloading method, there are two cases: single-hop and multi-hop service offloading. Therefore,  the computation delay can be presented as follows.
\begin{equation}
    T^p(t)=\sum_{i=1}^M \sum_{k=1}^K P_{c}t_{k}x_{i}^k(t)+H_k(t)P_{nc}t_{k}\left(1-x_{i}^k(t)\right)
\label{eq21}
\end{equation}

\section{LDSO: Dynamic Service Offloading Algorithm Based on Lyapunov}
	
To solve the cost minimization problem represented by Equation~(\ref{eq1}), the Lyapunov function cluster is employed to control the edge offloading system. And then, the time-dependent long-term optimization problem is transformed into an optimization problem just requiring current time slot information. Furthermore, a service matching algorithm based on greedy strategy is proposed to solve this optimization problem. Figure 3 shows the framework of LDSO (\emph{D}ynamic \emph{S}ervice \emph{O}ffloading based on \emph{L}yapunov). 

\subsection{Framework of Lyapunov Optimization Problem}
	
We discuss the quantitative representation of the system cost with the following definitions and lemmas in this section.

\begin{definition}[Quadratic Lyapunov Function]
    We have the following expression to denote the quadratic Lyapunov function.
\begin{equation}
    L(t) \stackrel{\text { def }}{=} \frac{1}{2} \sum_{i=1}^M\left(Q_{i}(t)\right)^2
\label{eq22}
\end{equation}
where $L(t)$ represents the buffer queue of the $i$-th edge server in time slot $t$, which is a scalar measure of queue congestion in the network. 
The smaller $L(t)$ is, the less the buffered data. 
It is obvious that $L(t)$ is equal to 0 if and only if all queues are empty, \emph{i.e.,} $Q_{i}(t)=0, \forall i \in[1, M]$. Otherwise, $L(t)$ is always greater than zero.
\end{definition}

\begin{definition}[Lyapunov Drift Function]
     The Lyapunov drift function is expressed as follows:
\begin{align}
    \Delta L(t) & \stackrel{\Delta} {=} E\{L(t+1)-L(t) \mid Q(t)\} \\
    & \stackrel{\Delta} {=} \frac{1}{2} \sum_{i=1}^M\left(Q_{i}(t+1)\right)^2-\frac{1}{2} \sum_{i=1}^M\left(Q_{i}(t)\right)^2
    \label{eq23}
\end{align} 
where $\Delta L(t)$ represents the change of buffered data in the queue. Based on the two definitions, we have the following Lemma 1.
\end{definition}
    
\begin{lemma}
    The upper bound on the variation of the buffered data in the edge server can be expressed as
\begin{equation}
    \Delta L(t) \leq B+E\left\{\sum_{i=1}^M Q_{i}(t)\left(A_{i}(t)-\mu_{i}(t)\right)\right\}.
\label{eq24}
\end{equation}
where
\begin{equation}
    B=\sum_{i=1}^M \frac{\left(\mu_{i}^{\max }\right)^2+\left(A_{i}^{\max }\right)^2}{2}
\label{eq25}
\end{equation}
\end{lemma}

\begin{proof}
    In each time slot, the amount of data offloaded by the edge server is always greater than or equal to zero, \emph{i.e.,}
\begin{equation}
    \max \left\{Q_{i}(t)-\mu_{i}(t), 0\right\} \leq Q_{i}(t)
\label{eq26}
\end{equation}
\begin{equation}
    \left(\max \left\{Q_{i}(t)-\mu_{i}(t), 0\right\}\right)^2 \leq\left(Q_{i}(t)-\mu_{i}(t)\right)^2
\label{eq27}
\end{equation}
Then Equation~(\ref{eq12}) can be rewrite as:
\begin{equation}
    \begin{aligned}
        \left(Q_{i}(t+1)\right)^2= & \left(\max \left\{Q_{i}(t)-\mu_{i}(t), 0\right\}\right)^2+\left(A_{i}(t)\right)^2 \\
	& +2 \max \left\{Q_{i}(t)-\mu_{i}(t), 0\right\}A_{i}(t) \\
	\leq & \left(Q_{i}(t)-\mu_{i}(t)\right)^2+\left(A_{i}(t)\right)^2 \\
	& +2 Q_{i}(t)A_{i}(t) \\
	= & \left(Q_{i}(t)\right)^2+\left(\mu_{i}(t)\right)^2+\left(A_{i}(t)\right)^2 \\
	& -2 Q_{i}(t)\left(\mu_{i}(t)-A_{i}(t)\right)
    \end{aligned}
\label{eq28}
\end{equation}
The Lyapunov drift function at time slot $t$ is then changed as follows.
\begin{equation}
    \begin{aligned}
        \Delta L(t)= & E\left\{\sum_{i=1}^M \frac{1}{2}\left(\left(Q_{i}(t+1)\right)^2-\left(Q_{i}(t)\right)^2\right)\right\} \\
	\leq & \sum_{i=1}^M\frac{\left(\mu_{i}^{\max}\right)^2+\left(A_{i}^{\max }\right)^2}{2} \\
        & +E\left\{\sum_{i=1}^M \left\{Q_{i}(t)\left(A_{i}(t)-\mu_{i}(t)\right)\right\} \right\}
    \end{aligned}
\label{eq29}
\end{equation}
\end{proof}
	
	
	
	
	

\begin{definition}[Lyapunov Drift-Plus-Penalty Function]
    \begin{equation}
    L_{\operatorname{cost}}(t)=\Delta L(t)+V \times E\{\operatorname{Cost}(t) \mid Q(t)\}
\label{eq30}
\end{equation}
where $V$ is a non-negative control parameter and $\operatorname{Cost}(t)$ is the system cost defined in Equation~(\ref{eq1}). By dynamically adjusting the control parameters $V$, the balance between system cost and buffered data can be achieved.
\end{definition}
	
The objective function $\operatorname{Cost}(t)$ is minimized based on stability. Therefore, $L_{\operatorname{cost}}(t)$ represents the offloading cost of the optimized time slot system. $\bigtriangleup L(t)$, as known from Lemma 1, measures the stability of the system. We have the following Lemma 2.

\begin{lemma}
    The upper bound on the system cost $L_{\operatorname{cost}}(t)$ is 
\begin{equation}
    L_{\operatorname{cost}}(t) \leq \widetilde{{L_{\operatorname{cost}}} }(t),
\label{eq31}
\end{equation}
where $\widetilde{{L_{\operatorname{cost}}}}(t)=B+V \cdot \operatorname{Cost}(t)+\sum_{i=1}^M Q_{i}(t)A_{i}(t)$ and denotes the upper bound of $L_{cost} (t)$ with the Lyapunov penalty.
\end{lemma}
	
\begin{proof}
    Substitute the result of Lemma 1 into Equation~(\ref{eq29}), the Drift-Plus-Penalty function $L_{\operatorname{cost}}(t)$ is scaled as follows.
\begin{equation}
    \begin{aligned}
        & L_{\operatorname{cost}}(t)=\Delta L(t)+V \cdot E\{\operatorname{Cost} (t) \mid Q(t)\} \\
	& \leq B+E\left\{\sum_{I=1}^M Q_{i}(t)\left(A_{i}(t)-\mu_{i}(t)\right)\right\} \\
	& +V \cdot E\{\operatorname{Cost}(t) \mid Q(t)\} \\
	& =B+V \cdot \operatorname{Cost}(t)+E\left\{\sum_{i=1}^M\left(Q_{i}(t) A_{i}(t)\right)\right\} \\
	& \quad-E\left\{\sum_{i=1}^M Q_{i}(t) \mu_{i}(t)\right\} \\
	& \leq B+V \cdot \operatorname{Cost}(t)+E\left\{\sum_{i=1}^M\left(Q_{i}(t) A_{i}(t)\right)\right\}
    \end{aligned}
\label{eq32}
\end{equation}
\end{proof}	

According to the framework of opportunistic minimization of conditional expectations [13], the real-time service offloading optimization problem is solved by minimizing the solution of the system cost. Let $C_{i}^{k}(t)$ represent the cost of the $i$-th edge server processing the $k$-th service with the Lyapunov Drift-Plus-Penalty function. We give the following theorem.

\begin{theorem}
    The real-time offloading cost of the $i$-th edge server can be expressed as follows.
\begin{equation}
    C_{i}^{k}(t)=V\theta E_{i}^{k}(t)+V\left(1-\theta\right)T_{i}^{k}(t)+Q_{i}^k(t)A_{i}^{k}(t)
\label{eq33}
\end{equation}
The two terms $E_{i}^{k}(t)$ and $T_{i}^{k}(t)$ denote the energy consumption and delay for the $k$-th service, respectively.
\end{theorem}

\begin{proof}
    Recall that $V \times \operatorname{Cost}(t)=V \theta E_{total}(t)+(V-V \theta) T_{total }(t)$ and the presentation of $\widetilde{{L_{\operatorname{cost}}}  } (t)$. We have
\begin{equation}
    \begin{gathered}
    L_{\operatorname{cost}}(t)=  V \times \operatorname{Cost}(t)+\sum_{i=1}^M\sum_{k=1}^K Q_{i}^k(t)A_{i}^{k}(t) \\
    =\sum_{i=1}^M\sum_{k=1}^K V\theta E_{i}^{k}(t)+V\left(1-\theta\right)T_{i}^{k}(t)+Q_{i}^k(t)A_{i}^{k}(t).
    \end{gathered}
\label{eq37}
\end{equation}
\end{proof}

Both the energy consumption and delay models include the binary variable $x_{{i}}^{k}(t)$. The key step is to determine the offloading point, \emph{i.e.,} where the value of $x_{{i}}^{k}(t)$ is equal to 1. By introducing Lyapunov optimization into the LDSO scheme, the dynamic service offloading optimization problem can be transformed into the matching problem of optimal parameters, which reduces the complexity of the LDSO scheme. Let $X(t)$ denote the matching matrix of $x_{i}^{k}(t)$. We present and discuss the implementation of LDSO in Algorithm 1.
	
\subsection{Matching Decision based on Greedy Strategy}
This study adopts a greedy idea to approximately obtain the solution of $x_{{i}}^{k}(t)$. 
The number of user requests, the service ratio, and the transmission rate are calculated in lines 4-10.  
Let $C_{i^{*} }^{k^{*}}(t)$ denote the smallest one of the set $\left\{C_{i}^{k}(t)\right\}_{M \times K}$. 
If the queue length is satisfied and the node caches the $k$-th service, the matching parameter $x_{{i}}^{k}(t)$ is set to 1 and $C_{i^{*} }^{k^{*}}(t)$ is removed from the set $\left\{C_{i}^{k}(t)\right\}_{M \times K}$ (lines 12-16). 
Then other elements excluding $i$ will be reset (lines 18-19). 
Finally, update the system cost and the buffer queue status (lines 25-26).

\makeatletter
\def\algocf@startloop{\begin{ALC@loop}}
\def\algocf@endloop{\end{ALC@loop}}
\makeatother

\begin{algorithm}[h!]
\SetAlgoLined 
\SetNoFillComment 
\DontPrintSemicolon 

\caption{LDSO.}\label{alg1}
\KwIn{control parameter $V$, weight factor $\theta $, $b_k$ and maximum queue length $Q_{i}^{\max}$}
\KwOut{matching parameter $x_{i}^{k}(t)$}

\BlankLine
$Q_{i}(t)=0,\ \mu_{i}(t)=0,\ A_{i}(t)=0;$\;

\For{$t=0$ \KwTo $T-1$}{
    $X(t) \leftarrow  {[-1]}_{M\times K}$ \tcp*[r]{Offloading Decision Matrix}
    \For{$k \in[1, K]$}{	
        Getting $n_{i}^{k}(t)$\;
    }
    \For{$i \in[1, M]$}{
        Getting $\mu_{i}(t)$ \tcp*[r]{service ratio}\;
    }
    Obtaining $r(t)$ based on $h(t)$\;
    
    \While{$\left\{C_{i}^{k}(t)\right\}_{M \times K} \neq \emptyset$}{
        $C_{i^{*} }^{k^{*}}(t) \leftarrow \min \{C_{i}^{k}(t)_{M\times K}\}$ \tcp*[r]{System cost function}\;
        
        \If {$(Q_{i}(t)+A_{i}^{k}(t) \le Q_{i}^{\max}) \land (X_{i^{*} }^{k^{*}}(t) \neq -1)$}{
            $n_{i}^{k}(t)=n_{i}^{k}(t-1)$\;
            $C_{i}^{k}(t)_{M\times K} \leftarrow C_{i}^{k}(t)_{M\times K}-C_{i^{*} }^{k^{*}}(t)$\;
            $X_{i^{*} }^{k^{*}}(t) \leftarrow 1$\;
            
            \For{$i \neq i^{*} \land i \in I$}{
                $X_{i^{*} }^{k^{*}}(t) \leftarrow 0$\;
                $C_{i}^{k}(t)_{M\times K} \leftarrow C_{i}^{k}(t)_{M\times K}-C_{i}^{k^{*}}(t)$\;
            }
        }
        \Else{
            $X_{i^{*} }^{k^{*}}(t) \leftarrow 0$\;
        }
    }
    
    $\operatorname{Cost}(t)=\theta \cdot E_{\text{total}}(t)+(1-\theta )T_{\text{total}}(t)$\;
    $Q_{i}(t+1) \leftarrow \max \{Q_{i}(t)-\mu_{i}(t)\}+A_{i}(t)$\;
}

\end{algorithm}

According to the convergence theory of Lyapunov optimization~\citep{r13} and the supporting hyperplane theorem~\citep{r52}, the following convergence of LDSO algorithm can be obtained. Define $\varepsilon=1 / V(\varepsilon>0)$, and for $T \geq 1 / \varepsilon^2$, we have the following theorem.

\begin{theorem}
    LDSO provides an $O(\varepsilon)$ approximation with convergence time $O\left(1 / \varepsilon^2\right)$ if the following two conditions are satisfied.
\begin{equation}
    \begin{gathered}
    \bar{C}(t) \leq C^*+O(\varepsilon) \\
    \overline{A_{i}}(t) \leq \mu+O(\varepsilon) \quad
    \end{gathered}
\label{eq34}
\end{equation}
where $C^*$ is the approximate optimal solution. 
\end{theorem}

\begin{proof}
According to Lyapunov optimization~\cite{r13} and convergence theory~\cite{r51}, the problem in this study is convex,  and an equivalent minimization problem can be constructed as follows.
\begin{equation}
    \begin{gathered}
        \min C \\
        \text { s.t. } A_{i} \leq \mu_{i}(t)\leq {\mu_{i}}^{max} \\
        \left(C, A_1, A_2, \ldots, A_M\right) \in \bar{R}
        \end{gathered}
\label{eq38}
\end{equation}
where $C=E\left \{ \operatorname{Cost}(t) \right \}$  and $A_{i}=E\left \{ A_{{i}}  (t) \right \} $  and $\bar{R} $ is the closed package of the set. From the supporting hyperplane theorem~\cite{r52}, there exists a hyperplane through the closed point $\left. ( C,\mu,\mu,...,\mu)\right.$ of $\bar{R}$, where $\mu$ is the maximum among $\mu_{i}^{max}, i \in [1,M]$. Also, there exists a non-negative normal vector $\Upsilon =\left ( \gamma _{0} ,\gamma_{ 1} ,...,\gamma _{M} \right ) $ supporting the hyperplane, satisfying the following inequality.
\begin{equation}
    \gamma_0 C+\sum_{i \in [1,M]} \gamma_i A_i \geq \gamma_0 \times C^*+\sum_{i \in [1,M]} \gamma_i \mu
\label{eq39}
\end{equation}

When $\gamma _{0} \ne 0$, the supporting hyperplane is non-perpendicular. By dividing the above inequality by $\gamma _{0}$ and defining $\delta _{i}=\gamma _{i} /\gamma _{0} $, Equation~(\ref{eq39}) is changed as follows.
\begin{equation}
    C+\sum_{i \in [1,M]} \delta_i A_i \geq C^*+\sum_{i \in [1,M]} \delta_i \mu
\label{eq40}
\end{equation}
where $\left ( \delta _{1} ,\delta _{2} ,...,\delta _{M}  \right )$ is the Lagrange multiplier vector~\cite{r51}. Thus $\left(\bar{C}(t), \overline{A_1}(t), \overline{A_2}(t), \ldots, \overline{A_M}(t)\right) \in \bar{R}$ satisfies inequation~(\ref{eq41}) and we have
\begin{equation}
    \begin{gathered}
        \bar{C}(t)+\sum_{i \in [1,M]} \delta_i \overline{A_i}(t) \geq C^*+\sum_{i \in [1,M]} \delta_i \mu \\
        C^*-\bar{C}(t) \leq \sum_{i \in [1,M]} \delta_i\left(\overline{A_i}(t)-\mu\right).
    \end{gathered}
\label{eq41}
\end{equation}

According to the size of the queue length, the queue of edge servers satisfies the following inequality.
\begin{equation}
    Q_{i}(t+1) \geq Q_{i}(t)+A_{i}(t)-\mu_{i}(t)
\label{eq42}
\end{equation}

Summing over all time slots $t\in \left \{ 0,1,2,...,T-1 \right \} $ from both sides of the Equation~(\ref{eq42}) and according to the law of telescoping sum, it is derived as follows.
\begin{equation}
    \begin{aligned}
        & \left(Q_{i}(T)-Q_{i}(T-1)+\ldots .+Q_{i}(1)-Q_{i}(0)\right) \\
        & +\sum_{t=0}^{T-1} \mu_{i}(t) \geq \sum_{t=0}^{T-1} A_{i}(t)
    \end{aligned}
\label{eq43}
\end{equation}

Then, dividing both sides by $T$, we have
\begin{equation}
    \frac{1}{T}\left(Q_{i}(T)-Q_{i}(0)\right)+\frac{1}{T} \sum_{t=0}^{T-1} \mu_{i}(t) \geq \frac{1}{T} \sum_{t=0}^{T-1} A_{i}(t).
\label{eq44}
\end{equation}

Obtain the expectation on both sides of the inequality, taking $T\longrightarrow \infty $. The queue limitation satisfies the following constraint.
\begin{equation}
    \lim _{T \rightarrow \infty} \frac{1}{T} E\left\{Q_{i}(T)\right\}+\overline{\mu_{i}}(t) \geq \lim _{T \rightarrow \infty} \frac{1}{T} \sum_{t=0}^{T-1} A_{i}(t)
\label{eq45}
\end{equation}

Substituting Equation~(\ref{eq45}) to (\ref{eq41}), we have
\begin{equation}
    C^* \leq \bar{C}(t)+\sum_{i \in [1,M]} \delta_i \frac{E\left\{Q_{i}(T)\right\}}{T}.
\label{eq46}
\end{equation}
	
On the other hand, summing the telescoping series over $t \in\{0,1,2, \ldots, T-1\}$ and $\operatorname{Cost}(t)>\operatorname{C}_{l}$, we have
\begin{equation}
    \begin{aligned}
        & E\{L(T)\}-E\{L(0)\} \\
        & \leq T B+T V \times \operatorname{C}^*-V \sum_{t=0}^{T-1} E\{\operatorname{Cost}(t)\}. \\
        & \leq T B+T V \times \operatorname{C}^*-TV \operatorname{C}_{l}
        \end{aligned}
\label{eq47}
\end{equation}
	
Dividing both sides by $T$ and using the conclusion of Equation~(\ref{eq46}), we have	
\begin{equation}
    \begin{aligned}
        \frac{1}{T} E\{L(T)\} & \leq B+V\left(\mathrm{C}^*-\mathrm{C}_{l}\right) \\
        & \leq B+V\left(\mathrm{C}^*-\overline{\mathrm{C}}(t)\right). \\
        & \leq B+V \sum_{i \in [1,M]} \delta_i \frac{E\left\{Q_{i}(T)\right\}}{T} \\
    \end{aligned}
\label{eq48}
\end{equation}
	
\begin{equation}
    \frac{1}{T} E\{L(T)\} \leq B+\frac{V}{T}\|\delta\|\|E\{Q(T)\}\|
\label{eq49}
\end{equation}
where $\left \| * \right \|$ is the 2-norm of the vector and the dot product of two vectors is less than or equal to their norm product.

Using the Lyapunov-drift-function, we have
\begin{equation}
    \frac{E\left\{\|Q(T)\|^2\right\}}{2 T} \leq B+\frac{V}{T}\|\delta\|\|E\{Q(T)\}\|.
\label{eq50}
\end{equation}
	
Since $E\left\{\|Q(T)\|^2\right\} \geq\|E\{Q(T)\}\|^2$, we get
\begin{equation}
    \frac{1}{2 T}\|E\{Q(T)\}\|^2 \leq B+\frac{V}{T}\|\delta\|\|E\{Q(T)\}\|.
\label{eq51}
\end{equation}

After the simple algebras, we get
\begin{equation}
    2\|E\{Q(T)\}\|^2-2 V\|\delta\|\|E\{Q(T)\}\|-2 B T \leq 0.
\label{eq52}
\end{equation}
	
From the quadratic equation, we have
\begin{equation}
    \|E\{Q(T)\}\| \leq V\|\delta\|+\sqrt{V^2\left\|\delta\right\|^2+2 B T}.
\label{eq53}
\end{equation}

Recall the Equation~(\ref{eq45}), we have
\begin{equation}
    \begin{aligned}
        \overline{A_{i}}(t) & \leq \overline{\mu_{i}}(t)+\frac{E\left\{Q_{i}(T)\right\}}{T} \\
        & \leq \mu+\frac{V\|\delta\|+\sqrt{V^2\left\|\delta\right\|^2+2 B T}}{T}.
    \end{aligned}
\label{eq54}
\end{equation}
	
It can further be derived as follows if  $\varepsilon$ takes $1/V$ and $T>1/\varepsilon^{2}$.
\begin{equation}
    \begin{aligned}
        & \frac{V\|\delta\|+\sqrt{V^2\left\|\delta\right\|^2+2 B T}}{T} \\
        & =\frac{\|\delta\|}{\varepsilon T}+\sqrt{\frac{1}{\varepsilon^2 T^2}\left\|\delta\right\|^2+\frac{2 B}{T}} \\
        & \leq\|\delta\| \varepsilon+\sqrt{\varepsilon^2\left\|\delta\right\|^2+2 B \varepsilon^2} \\
        & =\|\delta\| \varepsilon+\varepsilon \sqrt{\left\|\delta\right\|^2+2 B} \\
        & =O(\varepsilon)
    \end{aligned}
\label{eq55}
\end{equation}
	
Equation~(\ref{eq54}) is derived as follows.
\begin{equation}
    \overline{A_{i}}(t) \leq \mu+\|\delta\| \varepsilon+\varepsilon \sqrt{\left\|\delta\right\|^2+2 B} \leq \mu+O(\varepsilon)
\label{eq56}
\end{equation}
	
On the other hand, from Equation~(\ref{eq46}), we have
\begin{equation}
    \bar{C}(t) \leq C^*+B \cdot \varepsilon \leq C^*+O(\varepsilon).
\label{eq57}
\end{equation}
	
Therefore, the proposed algorithm can satisfy the convergence constraints in Equation~(\ref{eq38}), and the algorithm's convergence is also proven.
\end{proof}

\subsection{System stability analysis}
In this section, we analyze the stability of the LDSO algorithm for edge offloading systems. According to the Lyapunov stability theorem in queuing theory~\citep{r13}, when the system reaches a steady state, the following theorem can be drawn.

\begin{theorem}
    The offloading performance remains stable if the system cost satisfies the following inequality.
\begin{equation}
    \lim _{T \rightarrow \infty} \sup \frac{1}{T} \sum_{t=0}^{T-1} E\{\operatorname{Cost}(t)\} \leq \operatorname{C}^*+\frac{B}{V}
\label{eq35}
\end{equation}
	
The Lyapunov  \textbf{\emph{Drift-Plus-Penalty}} function has a stable feature, which can use the control parameter $V$ to regulate the trade-off between the average offloading cost and the queue length. The LOSO algorithm thus reaches a strong stability of the buffer queue.
\end{theorem}

\begin{proof}
    Based on Equation~(\ref{eq47}), we get
\begin{equation}
    \begin{gathered}
        E\{L(T)\}-E\{L(0)\}+V \sum_{t=0}^{T-1} E\{\operatorname{Cost}(t)\} \\
        \leq T\cdot B+T\cdot V \times \operatorname{C}^*
    \end{gathered}
\label{eq58}
\end{equation}

Dividing by $V \cdot T$ in both sides, we get
\begin{equation}
    \frac{E\{L(T)\}-E\{L(0)\}}{V\cdot T}+\frac{\sum_{t=0}^{T-1} E\{\operatorname{Cost}(t)\}}{T} \leq \frac{B}{V}+\operatorname{C}^*.
\label{eq59}
\end{equation}
After some simple algebra, we get
\begin{equation}
    \begin{gathered}
        \frac{1}{T} \sum_{t=0}^{T-1} E\{\operatorname{Cost}(t)\} \leq \frac{B}{V}+\operatorname{C}^*-\frac{E\{L(T)\}-E\{L(0)\}}{V\cdot T}\\
        \frac{1}{T} \sum_{t=0}^{T-1} E\{\operatorname{Cost}(t)\} \leq \frac{B}{V}+\operatorname{C}^*.
    \end{gathered}
\label{eq60}
\end{equation}

\end{proof}

\begin{theorem}
    The buffer queue of edge servers has strong stability and satisfies the following constraints.
\begin{equation}
    \lim _{T \rightarrow \infty} \sup \frac{1}{T^2} \sum_{t=0}^{T-1} \sum_{i=1}^M \frac{1}{2}\left\{Q_{i}(t)\right\}^2 \leq B+V(\operatorname{C}^*-\operatorname{C}_{l}) 
\label{eq36}
\end{equation}
where $\operatorname{C}_{l}$ is the lower bound of the expected penalty function. Strong stability indicates that the time-averaged buffered data has an upper bound, which strictly limits how fast the queue can grow.
\end{theorem}

\begin{proof}
    Recall that the result of Equation~(\ref{eq47}) $E\{L(T)\}-E\{L(0)\}  \leq T\cdot B+T\cdot V \times \operatorname{C}^*-T\cdot V \operatorname{C}_{l}$ and the Lyapunov function $L(t) \stackrel{\text { def }}{=} \frac{1}{2} \sum_{i=1}^M\left(Q_{i}(t)\right)^2$, we have

\begin{equation}
    \begin{aligned}
        & \frac{1}{T} \sum_{t=0}^{T-1} \sum_{i=1}^M \frac{1}{2}\left\{Q_{i}(t)\right\}^2-\frac{1}{T} \sum_{t=0}^{T-1} \sum_{i=1}^M \frac{1}{2}\left\{Q_{i}(0)\right\}^2 \\
        & \leq T\cdot B+T\cdot V (\operatorname{C}^*- \operatorname{C}_{l})\\
        &= T(B+V(\operatorname{C}^*-\operatorname{C}_{l})).
    \end{aligned}
\label{eq61}
\end{equation}
	
Let $Q_{i}(0)=0$, and the strong stability of the queue is proven in Theorem 4.
\end{proof}

\vspace{-0.4cm}
\section{Performance simulation}
\subsection{Simulation parameter setting}
\begin{figure}[htbp]
 \centering
   \includegraphics[width=0.4\textwidth]{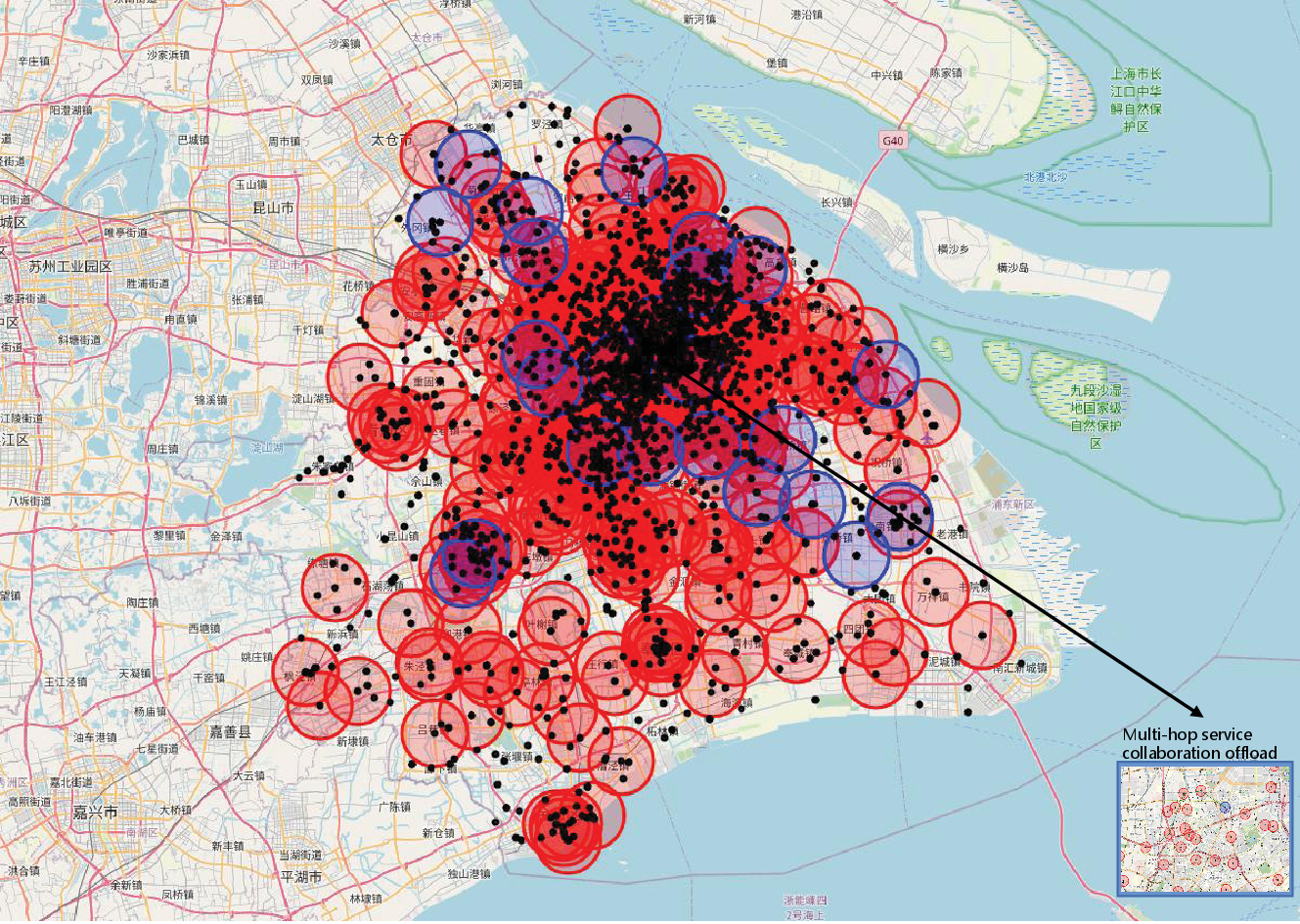}
   \caption{Deployment of edge servers}
   \label{Figure 4}
   \vspace{0.5cm}
\end{figure}
This study uses the Shanghai Telecom dataset to simulate an edge service offloading system as shown in Figure 4, where SBSs with blue color represent those that store the $k$-th type of service and red ones indicate those without caching this type. In the initial settings, the number of service types is equal to 10. 
Each SBS covers 50 users, and the communication range between users and SBSs is 15m, and the communication range between SBSs is 30m. A total of 1000 time slots are simulated, each with a duration of 1ms. The bandwidth of the link is 40MHz, and the channel gain is set to 20dB. 
The transmission power of each IoT device is set to $P_{u} \sim U[0.01, 1] $ W. In the system stability evaluation, the impact of the control parameter $V$ on the buffer queue is evaluated, in which the time-varying service request obeys the Poisson distribution with an average value of 20, and the weight factor $\theta$ is set to 0.5.

\subsection{Performance Evaluation}
\vspace{-0.1cm}	
\subsubsection{Algorithm Convergence Analysis}
This subsection analyzes the convergence of LDSO based on Figures \ref{fig:stability_v} and Figure \ref{fig:stability_lambda}, under different control parameters $V$ and service arrival rates $\lambda$. Three key observations emerge: (1) $V$ and $\lambda$ affect convergence differently, $V$ enables fine-grained adjustment, while $\lambda$ exerts a coarse-grained influence. LDSO stabilizes within 50 and 100 time slots for $V=1300$ and $V=5000$, but requires 350 and 1200 time slots for $\lambda=15$ and $\lambda=25$, respectively; (2) LDSO reaches stability quickly, validating the correctness and efficiency of Theorem 2; (3) Convergence slows as $V$ and $\lambda$ increase, due to larger data buffers requiring more time to process.

\begin{figure}[htbp]
\vspace{-0.4cm}	
    \centering
    \subfigure[System stability with different $V$]{
        \includegraphics[width=0.23\textwidth]{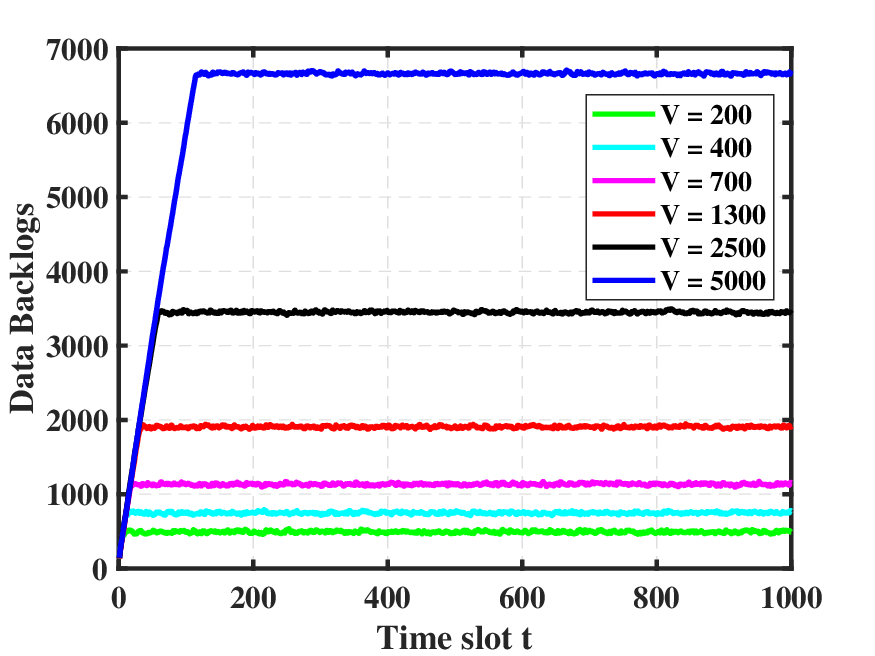}
        \label{fig:stability_v}
    }
    \hfill
    \subfigure[System stability with different $\lambda$]{
        \includegraphics[width=0.23\textwidth]{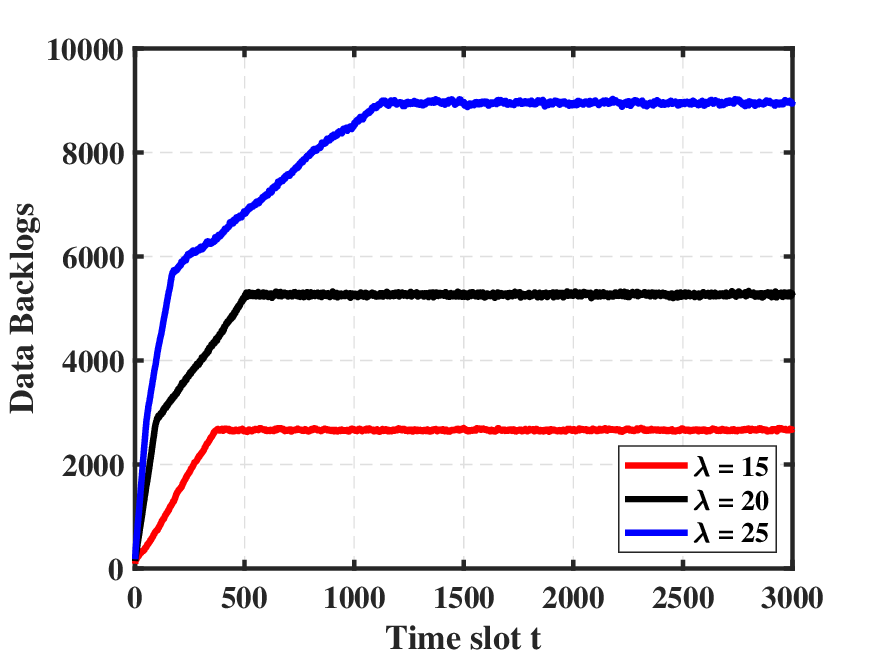}
        \label{fig:stability_lambda}
    }
    \caption{System stability analysis for different $V$ and $\lambda$.}
    \label{fig:performance}
    \vspace{0.3cm}
\end{figure}




    


\subsubsection{The influence of the control parameter V on the system cost}

\vspace{-0.8cm}
\begin{figure}[htbp]
    \centering
    \subfigure[System cost \emph{w/} $Q_i^{\max} = 4000$]{
        \includegraphics[width=0.23\textwidth]{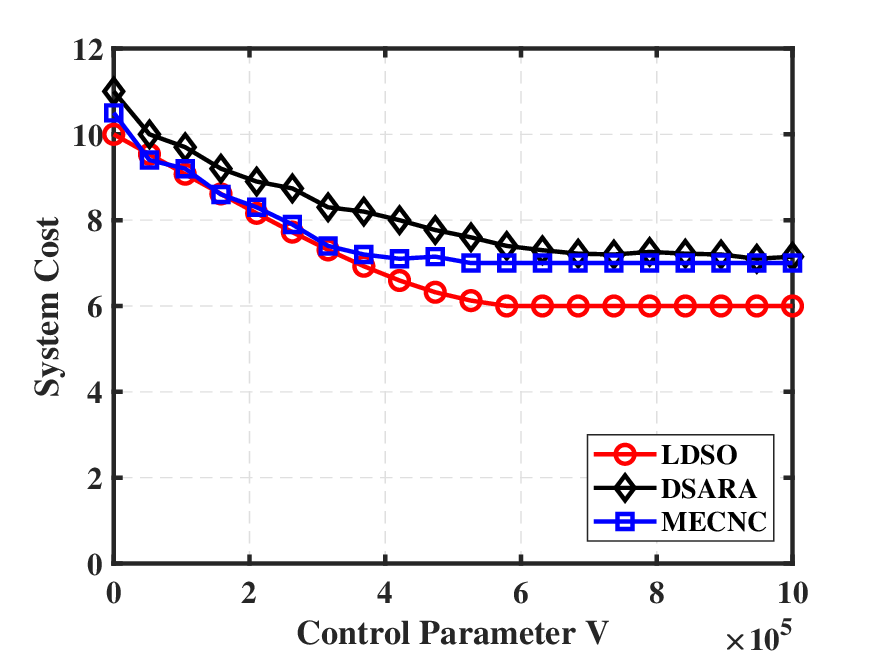}
        \label{fig:cost_with_bound}
    }
    \hfill
    \subfigure[System cost \emph{w/o} $Q_i^{\max} = 4000$]{
        \includegraphics[width=0.23\textwidth]{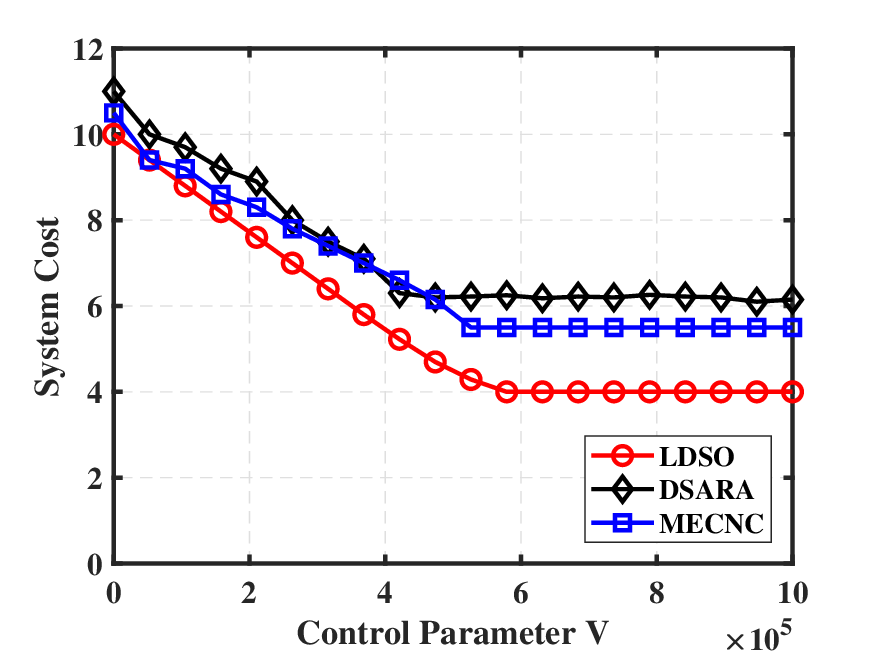}
        \label{fig:cost_without_bound}
    }
    \caption{Performance comparison of LDSO, DSARA, and MECNC under different control parameters.}
    \label{fig:performance1}
    \vspace{0.3cm}
\end{figure}
We now compare LDSO with the state-of-the-art works. The comparison algorithms are DSARA~\citep{r41} and MECNC~\citep{r44}.
Firstly, we analyze the system cost of the three algorithms. Figure \ref{fig:cost_with_bound} and Figure \ref{fig:cost_without_bound} plot the experimental results with/without the constraint of queue length, respectively (the queue length is bounded by 4000 in Figure \ref{fig:cost_with_bound}). It can be seen that the system cost shows a declining trend when the value of $V$ increases, which verifies the correctness of Theorem 3, \emph{i.e.,} a bigger $V$ means a smaller cost as shown in Equation \ref{eq35}). Specifically, LDSO achieves the smallest system cost compared with the DSARA and MECNC. The reduced cost is about an average of $10\%$ in Figure \ref{fig:cost_with_bound}, and more improvement can be achieved in Figure \ref{fig:cost_without_bound} if the queue length is unlimited.

\subsubsection{The influence of the control parameter V on the volumes of buffered data}

\begin{figure}[htbp]
\vspace{-0.8cm}
    \centering
    \subfigure[System cost \emph{w/} $Q_i^{\max} = 4000$]{
        \includegraphics[width=0.23\textwidth]{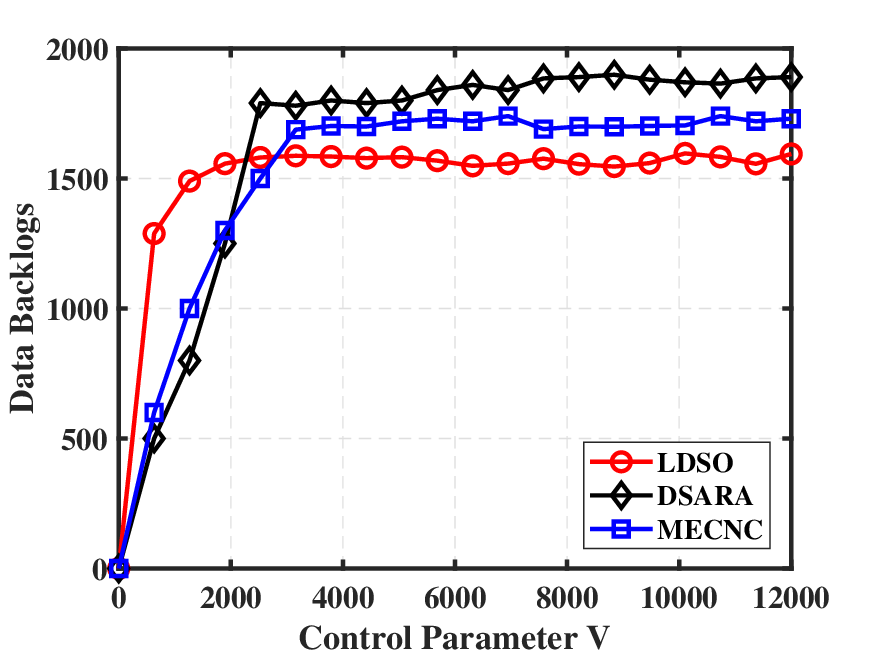}
        \label{fig:fig6}
    }
    \hfill
    \subfigure[System cost \emph{w/o} $Q_i^{\max} = 4000$]{
        \includegraphics[width=0.23\textwidth]{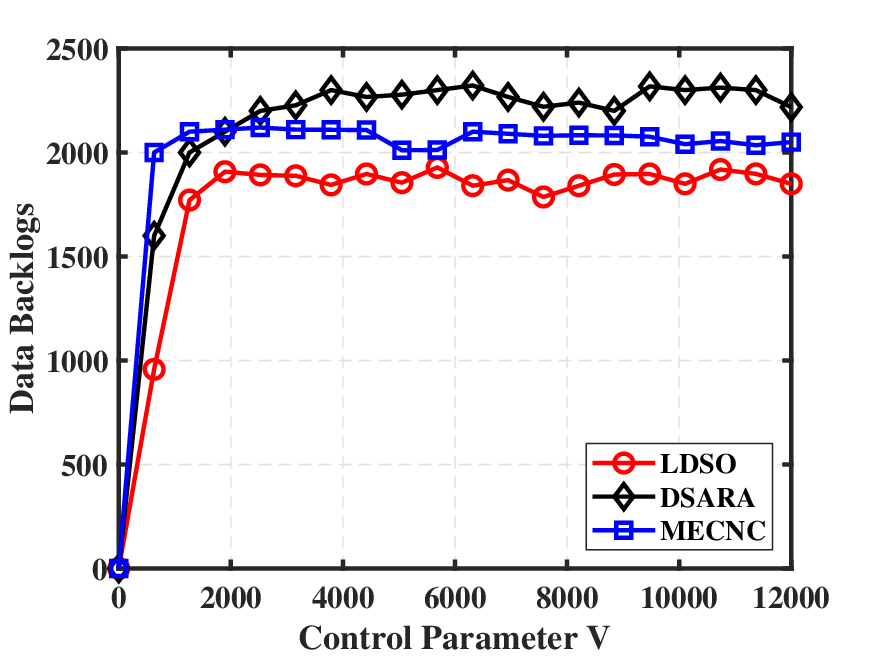}
        \label{fig:fig7}
    }
    \subfigure[Long period system cost]{
        \includegraphics[width=0.23\textwidth]{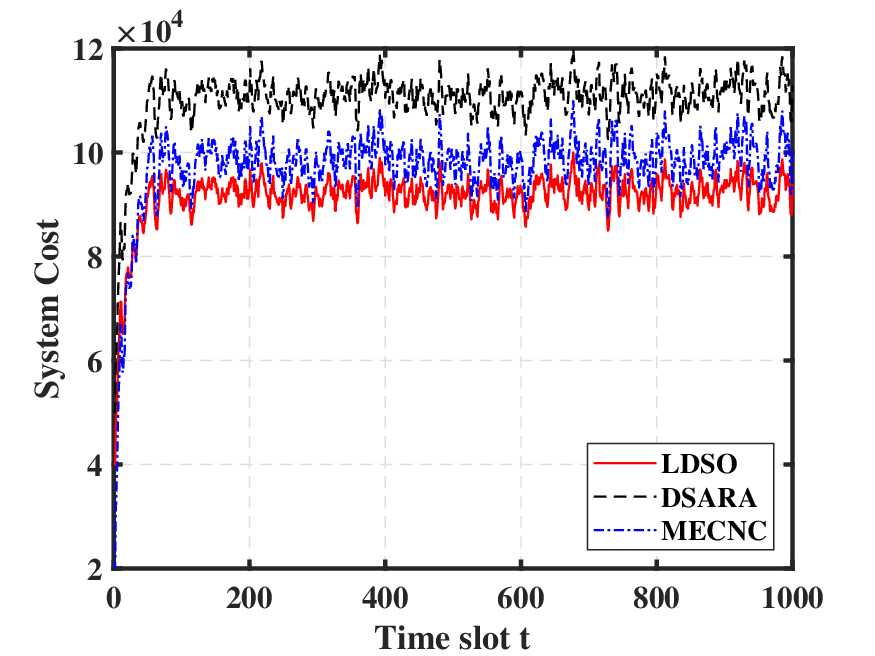}
        \label{fig:fig8}
    }
    \hfill
    \subfigure[Long period data offloading]{
        \includegraphics[width=0.23\textwidth]{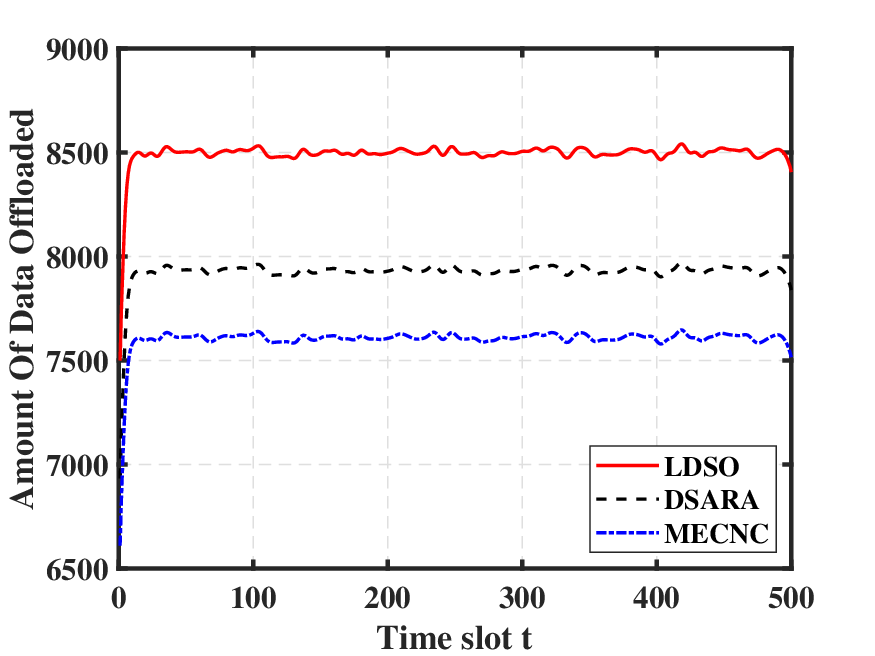}
        \label{fig:fig9}
    }
    \caption{Performance comparison of LDSO, DSARA, and MECNC under different control parameters over time.}
    \label{fig:system_cost_offloading}
    \vspace{0.5cm}
\end{figure}
The buffered data volumes of the three algorithms are analyzed, as shown in Figure \ref{fig:fig6} and \ref{fig:fig7}. Overall, the buffered data increases with larger $V$, consistent with Theorem 4 (Equation \ref{eq36}), and a higher $V$ allows more data to be cached. Among the three, LDSO consistently maintains the smallest buffer size, around 1600 in Figure \ref{fig:fig6} and \ref{fig:fig7}, due to tighter control from queue constraints, resulting in lower buffer occupancy.
We evaluate the long-term system utilization of the three algorithms, as shown in Figures \ref{fig:fig8} and Figure \ref{fig:fig9}. As service processing begins, both system cost and offloaded data volume increase significantly. Figure \ref{fig:fig8} shows that LDSO achieves the lowest system cost, consistent with previous results, although the cost rises initially before stabilizing. Figure \ref{fig:fig9} shows that LDSO offloads the largest amount of data, with 8500 blocks compared to 7900 for DSARA and 7600 for MECNC. Overall, LDSO reduces system cost by 10 percent while increasing the offloaded data volume by 18.75 percent on average, indicating better efficiency and scalability.

\section{Conclusion}
This study has investigated the problem of cost minimization in a real-time edge offloading system with cooperation among edge servers. By utilizing the Lyapunov optimization framework, an online dynamic service offloading algorithm, LDSO, has been proposed to enable rational service decision-making. LDSO does not require prior knowledge and has achieved a trade-off between system stability and offloading cost. Intensive theoretical analysis and extensive experimental results have demonstrated the effectiveness and efficiency of LDSO. In future work, Lyapunov-based and intelligent service offloading principles will be further explored to enhance users’ quality of experience.
\section*{Acknowledgements}
This work is partially supported by the National Natural Science Foundation of China under Grant No.62072159 and No.61902112, the Science and Technology Research Project of Henan province under Grant No.222102210011, No.232102211061, and No.252102210218.

\vspace{-0.4cm}	
\bibliography{mybibfile}

\begin{thebibliography}{2}

\bibitem{duan2018treasure}
Duan, Huayi and Zheng, Yifeng and Wang, Cong and Yuan, Xingliang.
\newblock Treasure collection on foggy islands: Building secure network archives for Internet of Things.
\newblock IEEE Internet of Things Journal, 6(2), 2637–2650, 2018.

\bibitem{r2}
Yuan, Peiyan and Huang, Rong and Zhang, Junna and Zhang, En and Zhao, Xiaoyan.
\newblock Accuracy Rate Maximization in Edge Federated Learning With Delay and Energy Constraints.
\newblock IEEE Systems Journal, 17(2), 2053–2064, 2023.

\bibitem{r3}
Guo, Min and Wang, Wei and Huang, Xing and Chen, Yanru and Zhang, Lei and Chen, Liangyin.
\newblock Lyapunov-based partial computation offloading for multiple mobile devices enabled by harvested energy in MEC.
\newblock IEEE Internet of Things Journal, 9(11), 9025–9035, 2021.

\bibitem{r4}
Zhou, Bowen and Dastjerdi, Amir Vahid and Calheiros, Rodrigo N and Buyya, Rajkumar.
\newblock An online algorithm for task offloading in heterogeneous mobile clouds.
\newblock ACM Transactions on Internet Technology (TOIT), 18(2), 1–25, 2018.

\bibitem{r5}
Beloglazov, Anton and Buyya, Rajkumar and Lee, Young Choon and Zomaya, Albert.
\newblock A taxonomy and survey of energy-efficient data centers and cloud computing systems.
\newblock Advances in computers, 82, 47–111, 2011.

\bibitem{r6}
Yuan, Peiyan and Cai, Yunyun and Huang, Xiaoyan and Tang, Shaojie and Zhao, Xiaoyan.
\newblock Collaboration improves the capacity of mobile edge computing.
\newblock IEEE Internet of Things Journal, 6(6), 10610–10619, 2019.

\bibitem{r7}
Rahmani, Amir M and Gia, Tuan Nguyen and Negash, Behailu and Anzanpour, Arman and Azimi, Iman and Jiang, Mingzhe and Liljeberg, Pasi.
\newblock Exploiting smart e-Health gateways at the edge of healthcare Internet-of-Things: A fog computing approach.
\newblock Future Generation Computer Systems, 78, 641–658, 2018.

\bibitem{r8}
Zhao, Hailiang and Deng, Shuiguang and Zhang, Cheng and Du, Wei and He, Qiang and Yin, Jianwei.
\newblock A mobility-aware cross-edge computation offloading framework for partitionable applications.
\newblock In \emph{2019 IEEE International Conference on Web Services (ICWS)}, pp. 193–200, 2019.

\bibitem{r9}
Zyane, Abdellah and Bahiri, Mohamed Nabil and Ghammaz, Abdelilah.
\newblock IoTScal-H: hybrid monitoring solution based on cloud computing for autonomic middleware-level scalability management within IoT systems and different SLA traffic requirements.
\newblock International Journal of Communication Systems, 33(14), e4495, 2020.

\bibitem{r10}
Soyata, Tolga and Muraleedharan, Rajani and Funai, Colin and Kwon, Minseok and Heinzelman, Wendi.
\newblock Cloud-vision: Real-time face recognition using a mobile-cloudlet-cloud acceleration architecture.
\newblock In \emph{2012 IEEE symposium on computers and communications (ISCC)}, pp. 000059–000066, 2012.

\bibitem{r11}
Yuan, Peiyan and Shao, Saike and Geng, Lijuan and Zhao, Xiaoyan.
\newblock Caching hit ratio maximization in mobile edge computing with node cooperation.
\newblock Computer Networks, 200, 108507, 2021.

\bibitem{r12}
Vu, Thai T and Nguyen, Diep N and Hoang, Dinh Thai and Dutkiewicz, Eryk and Nguyen, Thuy V.
\newblock Optimal energy efficiency with delay constraints for multi-layer cooperative fog computing networks.
\newblock IEEE Transactions on Communications, 69(6), 3911–3929, 2021.

\bibitem{r13}
Neely, Michael J.
\newblock Stochastic network optimization with application to communication and queueing systems.
\newblock Synthesis Lectures on Communication Networks, 3(1), 1–211, 2010.

\bibitem{r14}
Xing, Hong and Liu, Liang and Xu, Jie and Nallanathan, Arumugam.
\newblock Joint task assignment and resource allocation for D2D-enabled mobile-edge computing.
\newblock IEEE Transactions on Communications, 67(6), 4193–4207, 2019.

\bibitem{r15}
Liu, Chen-Feng and Bennis, Mehdi and Debbah, Merouane and Poor, H Vincent.
\newblock Dynamic task offloading and resource allocation for ultra-reliable low-latency edge computing.
\newblock IEEE Transactions on Communications, 67(6), 4132–4150, 2019.

\bibitem{r16}
Tran, Tuyen X and Pompili, Dario.
\newblock Joint task offloading and resource allocation for multi-server mobile-edge computing networks.
\newblock IEEE Transactions on Vehicular Technology, 68(1), 856–868, 2018.

\bibitem{r17}
Mao, Yuyi and Zhang, Jun and Letaief, Khaled B.
\newblock Dynamic computation offloading for mobile-edge computing with energy harvesting devices.
\newblock IEEE Journal on Selected Areas in Communications, 34(12), 3590–3605, 2016.

\bibitem{r18}
Jo{\v{s}}ilo, Sla{\dj}ana and D{\'a}n, Gy{\"o}rgy.
\newblock Computation offloading scheduling for periodic tasks in mobile edge computing.
\newblock IEEE/ACM Transactions on Networking, 28(2), 667–680, 2020.

\bibitem{r19}
Chen, Ying and Zhang, Ning and Zhang, Yongchao and Chen, Xin and Wu, Wen and Shen, Xuemin.
\newblock Energy efficient dynamic offloading in mobile edge computing for internet of things.
\newblock IEEE Transactions on Cloud Computing, 9(3), 1050–1060, 2019.

\bibitem{r20}
Shu, Peng and Liu, Fangming and Jin, Hai and Chen, Min and Wen, Feng and Qu, Yupeng and Li, Bo.
\newblock eTime: Energy-efficient transmission between cloud and mobile devices.
\newblock In \emph{2013 Proceedings IEEE INFOCOM}, pp. 195–199, 2013.

\bibitem{r21}
Cui, Ying and Lau, Vincent KN and Wang, Rui and Huang, Huang and Zhang, Shunqing.
\newblock A survey on delay-aware resource control for wireless systems¡ªLarge deviation theory, stochastic Lyapunov drift, and distributed stochastic learning.
\newblock IEEE Transactions on Information Theory, 58(3), 1677–1701, 2012.

\bibitem{r22}
Sun, Yuxuan and Zhou, Sheng and Xu, Jie.
\newblock EMM: Energy-aware mobility management for mobile edge computing in ultra dense networks.
\newblock IEEE Journal on Selected Areas in Communications, 35(11), 2637–2646, 2017.

\bibitem{r23}
Chen, Lixing and Song, Linqi and Chakareski, Jacob and Xu, Jie.
\newblock Collaborative content placement among wireless edge caching stations with time-to-live cache.
\newblock IEEE transactions on multimedia, 22(2), 432–444, 2019.

\bibitem{r24}
Bahn, Hyokyung.
\newblock Efficiency of Buffer Caching in Computing-Intensive Workloads.
\newblock In \emph{2020 7th International Conference on Information Science and Control Engineering (ICISCE)}, pp. 548–552, 2020.

\bibitem{r25}
Wu, Wen and Yang, Peng and Zhang, Weiting and Zhou, Conghao and Shen, Xuemin.
\newblock Accuracy-guaranteed collaborative DNN inference in industrial IoT via deep reinforcement learning.
\newblock IEEE Transactions on Industrial Informatics, 17(7), 4988–4998, 2020.

\bibitem{r26}
Lin, Rongping and Zhou, Zhijie and Luo, Shan and Xiao, Yong and Wang, Xiong and Wang, Sheng and Zukerman, Moshe.
\newblock Distributed optimization for computation offloading in edge computing.
\newblock IEEE Transactions on Wireless Communications, 19(12), 8179–8194, 2020.

\bibitem{r27}
Zhang, Ni and Guo, Songtao and Dong, Yifan and Liu, Defang.
\newblock Joint task offloading and data caching in mobile edge computing networks.
\newblock Computer Networks, 182, 107446, 2020.

\bibitem{r28}
Chiang, Mung and Zhang, Tao.
\newblock Fog and IoT: An overview of research opportunities.
\newblock IEEE Internet of things journal, 3(6), 854–864, 2016.

\bibitem{r29}
Xu, Zichuan and Liang, Weifa and Xu, Wenzheng and Jia, Mike and Guo, Song.
\newblock Efficient algorithms for capacitated cloudlet placements.
\newblock IEEE Transactions on Parallel and Distributed Systems, 27(10), 2866–2880, 2015.

\bibitem{r30}
Lin, Hai and Zeadally, Sherali and Chen, Zhihong and Labiod, Houda and Wang, Lusheng.
\newblock A survey on computation offloading modeling for edge computing.
\newblock Journal of Network and Computer Applications, 169, 102781, 2020.

\bibitem{r31}
Zhang, Yao and Li, Changle and Luan, Tom Hao and Fu, Yuchuan and Shi, Weisong and Zhu, Lina.
\newblock A mobility-aware vehicular caching scheme in content centric networks: Model and optimization.
\newblock IEEE Transactions on Vehicular Technology, 68(4), 3100–3112, 2019.

\bibitem{r32}
Wu, Huaming and Wolter, Katinka and Jiao, Pengfei and Deng, Yingjun and Zhao, Yubin and Xu, Minxian.
\newblock EEDTO: an energy-efficient dynamic task offloading algorithm for blockchain-enabled IoT-edge-cloud orchestrated computing.
\newblock IEEE Internet of Things Journal, 8(4), 2163–2176, 2020.

\bibitem{r33}
Ouyang, Tao and Zhou, Zhi and Chen, Xu.
\newblock Follow me at the edge: Mobility-aware dynamic service placement for mobile edge computing.
\newblock IEEE Journal on Selected Areas in Communications, 36(10), 2333–2345, 2018.

\bibitem{r34}
Tang, Chaogang and Zhu, Chunsheng and Wu, Huaming and Liu, Chunyan and Rodrigues, Joel JPC.
\newblock Caching assisted correlated task offloading for IoT devices in mobile edge computing.
\newblock In \emph{2021 IEEE Global Communications Conference (GLOBECOM)}, pp. 1–6, 2021.

\bibitem{r35}
Li, Yuqing and Wang, Xiong and Gan, Xiaoying and Jin, Haiming and Fu, Luoyi and Wang, Xinbing.
\newblock Learning-aided computation offloading for trusted collaborative mobile edge computing.
\newblock IEEE Transactions on Mobile Computing, 19(12), 2833–2849, 2019.

\bibitem{r36}
Gao, Jixun and Chang, Rui and Yang, Zhipeng and Huang, Quanzheng and Zhao, Yuanyuan and Wu, Yu.
\newblock A task offloading algorithm for cloud-edge collaborative system based on Lyapunov optimization.
\newblock Cluster Computing, 26(1), 337–348, 2023.

\bibitem{r37}
Xia, Xiaoyu and Chen, Feifei and He, Qiang and Grundy, John and Abdelrazek, Mohamed and Jin, Hai.
\newblock Online collaborative data caching in edge computing.
\newblock IEEE Transactions on Parallel and Distributed Systems, 32(2), 281–294, 2020.

\bibitem{r38}
Chen, Xiangyi and Bi, Yuanguo and Chen, Xueping and Zhao, Hai and Cheng, Nan and Li, Fuliang and Cheng, Wenlin.
\newblock Dynamic Service Migration and Request Routing for Microservice in Multicell Mobile-Edge Computing.
\newblock IEEE Internet of Things Journal, 9(15), 13126–13143, 2022.

\bibitem{r39}
Duan, Xiaoting and Xu, Fei and Sun, Yongyong.
\newblock Research on offloading strategy in edge computing of internet of things.
\newblock In \emph{2020 International Conference on Computer Network, Electronic and Automation (ICCNEA)}, pp. 206–210, 2020.

\bibitem{r40}
Ning, Zhaolong and Dong, Peiran and Wang, Xiaojie and Wang, Shupeng and Hu, Xiping and Guo, Song and Qiu, Tie and Hu, Bin and Kwok, Ricky YK.
\newblock Distributed and dynamic service placement in pervasive edge computing networks.
\newblock IEEE Transactions on Parallel and Distributed Systems, 32(6), 1277–1292, 2020.

\bibitem{r41}
Feng, Jie and Liu, Lei and Pei, Qingqi and Hou, Fen and Yang, Tingting and Wu, Jinsong.
\newblock Service characteristics-oriented joint optimization of radio and computing resource allocation in mobile-edge computing.
\newblock IEEE Internet of Things Journal, 8(11), 9407–9421, 2021.

\bibitem{r42}
Pan, Meini and Li, Zhihua.
\newblock Multi-user Computation Offloading Algorithm for Mobile Edge Computing.
\newblock In \emph{2021 2nd International Conference on Electronics, Communications and Information Technology (CECIT)}, pp. 771–776, 2021.

\bibitem{r43}
Hu, Han and Song, Weiwei and Wang, Qun and Zhou, Fuhui and Hu, Rose Qingyang.
\newblock Mobility-aware offloading and resource allocation in MEC-enabled IoT networks.
\newblock In \emph{2020 16th International Conference on Mobility, Sensing and Networking (MSN)}, pp. 554–560, 2020.

\bibitem{r44}
Cai, Yang and Llorca, Jaime and Tulino, Antonia M and Molisch, Andreas F.
\newblock Mobile edge computing network control: Tradeoff between delay and cost.
\newblock In \emph{GLOBECOM 2020-2020 IEEE Global Communications Conference}, pp. 1–6, 2020.

\bibitem{r45}
Fan, Qingyang and Lin, Junyu and Feng, Guangsheng and Gao, Zihan and Wang, Huiqiang and Li, Yafei.
\newblock Joint service caching and computation offloading to maximize system profits in mobile edge-cloud computing.
\newblock In \emph{2020 16th international conference on mobility, sensing and networking (MSN)}, pp. 244–251, 2020.

\bibitem{r47}
Mao, Yuyi and Zhang, Jun and Song, SH and Letaief, Khaled Ben.
\newblock Power-delay tradeoff in multi-user mobile-edge computing systems.
\newblock In \emph{2016 IEEE global communications conference (GLOBECOM)}, pp. 1–6, 2016.

\bibitem{r49}
Hu, Han and Song, Weiwei and Wang, Qun and Hu, Rose Qingyang and Zhu, Hongbo.
\newblock Energy efficiency and delay tradeoff in an mec-enabled mobile iot network.
\newblock IEEE Internet of Things Journal, 9(17), 15942–15956, 2022.

\bibitem{r50}
Yuan, Peiyan and Li, Shuhong and Cai, Yunyun and Zhao, Xiaoyan and Tang, Shaojie and Li, Xiangyang.
\newblock Maximizing the Capacity of Edge-Caching Networks With User-Content Evolution Relationship.
\newblock IEEE Transactions on Vehicular Technology, 71(11), 12169–12178, 2022.

\bibitem{r51}
Neely, Michael J.
\newblock A simple convergence time analysis of drift-plus-penalty for stochastic optimization and convex programs.
\newblock arXiv preprint arXiv:1412.0791, 2014.

\bibitem{r52}
Bertsekas, Dimitri and Nedic, Angelia and Ozdaglar, Asuman.
\newblock Convex analysis and optimization.
\newblock Athena Scientific, 2003.

\bibitem{wang2024dependency}
Wang, Chenyang and Yu, Hao and Li, Xiuhua and Ma, Fei and Wang, Xiaofei and Taleb, Tarik and Leung, Victor CM.
\newblock Dependency-Aware Microservice Deployment for Edge Computing: A Deep Reinforcement Learning Approach with Network Representation.
\newblock IEEE Transactions on Mobile Computing, 2024.

\bibitem{wang2023heterogeneous}
Wang, Chenyang and Li, Ruibin and Wang, Xiaofei and Taleb, Tarik and Guo, Song and Sun, Yuxia and Leung, Victor CM.
\newblock Heterogeneous Edge Caching Based on Actor-Critic Learning With Attention Mechanism Aiding.
\newblock IEEE Transactions on Network Science and Engineering, 10(6), 3409–3420, 2023.

\bibitem{chen2022dynamic}
Chen, Yan and Sun, Yanjing and Wang, Chenyang and Taleb, Tarik.
\newblock Dynamic task allocation and service migration in edge-cloud iot system based on deep reinforcement learning.
\newblock IEEE Internet of Things Journal, 9(18), 16742–16757, 2022.

\bibitem{chen2021multitask}
Chen, Jiawen and Yang, Yajun and Wang, Chenyang and Zhang, Heng and Qiu, Chao and Wang, Xiaofei.
\newblock Multitask offloading strategy optimization based on directed acyclic graphs for edge computing.
\newblock IEEE Internet of Things Journal, 9(12), 9367–9378, 2021.

\bibitem{sun2024hierarchical}
Sun, Chuan and Li, Xiuhua and Wang, Chenyang and He, Qiang and Wang, Xiaofei and Leung, Victor CM.
\newblock Hierarchical Deep Reinforcement Learning for Joint Service Caching and Computation Offloading in Mobile Edge-Cloud Computing.
\newblock IEEE Transactions on Services Computing, 2024.

\bibitem{chen2024joint}
Chen, Yan and Sun, Yanjing and Yu, Hao and Taleb, Tarik.
\newblock Joint Task and Computing Resource Allocation in Distributed Edge Computing Systems via Multi-Agent Deep Reinforcement Learning.
\newblock IEEE Transactions on Network Science and Engineering, 2024.

\bibitem{wang2023mission}
Wang, Jingjing and Sun, Yanjing and Wang, Bowen and Ushio, Toshimitsu.
\newblock Mission-Aware UAV Deployment for Post-Disaster Scenarios: A Worst-Case SAC-Based Approach.
\newblock IEEE Transactions on Vehicular Technology, 2023.

\bibitem{yu2022deterministic}
Yu, Hao and Taleb, Tarik and Zhang, Jiawei.
\newblock Deterministic latency/jitter-aware service function chaining over beyond 5G edge fabric.
\newblock IEEE Transactions on Network and Service Management, 19(3), 2148–2162, 2022.

\bibitem{r10663201}
Wang, Chenyang and Yu, Hao and Li, Xiuhua and Ma, Fei and Wang, Xiaofei and Taleb, Tarik and Leung, Victor C. M..
\newblock Dependency-Aware Microservice Deployment for Edge Computing: A Deep Reinforcement Learning Approach With Network Representation.
\newblock IEEE Transactions on Mobile Computing, 23(12), 14737–14753, 2024.
\newblock DOI: 10.1109/TMC.2024.3453069.

\bibitem{r1234}
Wu, Tao and Li, Maomao and Qu, Yuben and Wang, Hongjun and Wei, Zhenhua and Cao, Jiannong.
\newblock Joint UAV Deployment and Edge Association for Energy-Efficient Federated Learning.
\newblock IEEE Transactions on Cognitive Communications and Networking, 1–1, 2025.
\newblock DOI: 10.1109/TCCN.2025.3543365.

\end{thebibliography}

\end{document}